\documentclass[11pt,a4paper]{article}

\title{\vspace{-2.5cm}\bf Noncommutative geometry, Lorentzian structures and causality}

\author{Nicolas Franco$^*$ \& Micha{\l} Eckstein$^{\dagger *}$}

\date{\footnotesize$^*$ Copernicus Center for Interdisciplinary Studies, \\
ul. S{\l}awkowska 17, 31-016 Krak\'ow, Poland \\[0.2cm]
 \hspace{-0,5cm}$^{\dagger}$ Jagiellonian University, Faculty of Mathematics and Computer Science,\hspace{-0.5cm}\\
 ul. {\L}ojasiewicza 6, 30-348 Krak\'ow, Poland \\[0.2cm]
  nicolas.franco@math.unamur.be \quad michal.eckstein@uj.edu.pl}

\usepackage[hyperfootnotes=false,hidelinks]{hyperref}
\usepackage[comma]{natbib}
\usepackage{har2nat}

\setcitestyle{square}
\bibpunct{[}{]}{,}{a}{}{,}

\usepackage[english]{babel} 
\usepackage{amsmath,amsfonts,amssymb,amsthm}
\usepackage{enumerate}
\usepackage{graphicx}

\newtheorem{theorem}{Theorem}

\newtheorem{proposition}[theorem]{Proposition}

\theoremstyle{definition}
\newtheorem{definition}[theorem]{Definition}

\newcommand{\norm}[1]{\left\Vert {#1} \right\Vert}
\newcommand{\abs}[1]{\left\vert {#1} \right\vert}
\newcommand{\set}[1]{\left\{ {#1} \right\}}

\newcommand{\scal}[1]{\left< {#1} \right>}

\newcommand{\setZ}{{\mathbb Z}}
\newcommand{\setR}{{\mathbb R}}
\newcommand{\setC}{{\mathbb C}}

\newcommand{\A}{\mathcal{A}}

\newcommand{\D}{\mathcal{D}}

\renewcommand{\H}{\mathcal{H}}

\newcommand{\J}{\mathcal{J}}
\newcommand{\M}{\mathcal{M}}

\renewcommand{\P}{\mathcal{P}}

\newcommand{\T}{\mathcal{T}}
\newcommand{\C}{\mathcal{C}}

\DeclareMathOperator{\Span}{span}

\DeclareMathOperator{\diag}{diag}

\DeclareMathOperator{\Tr}{Tr}

\begin{document}

\maketitle
\thispagestyle{empty}


\begin{abstract}
The theory of noncommutative geometry provides an interesting mathematical background for developing new physical models. In particular, it allows one to describe the classical Standard Model coupled to Euclidean gravity. However, noncommutative geometry has mainly been developed using the Euclidean signature, and the typical Lorentzian aspects of space-time, the causal structure in particular, are not taken into account. We present an extension of noncommutative geometry \`a la Connes suitable the for accommodation of Lorentzian structures. In this context, we show that it is possible to recover the notion of causality from purely algebraic data. We explore the causal structure of a simple toy model based on an almost commutative geometry and we show that the coupling between the space-time and an internal noncommutative space establishes a new `speed of light constraint'.
\end{abstract}
%
%

\section{Introduction}

The idea of combining mathematics with physics in order to describe the enclosing World comes from the ancient Greeks with the pioneering works of Archi\-medes, which can be found in his recently discovered palimpsest. Since then, from the first simple models by Ptolemy and Copernicus, the laws of the Universe have been successfully described by more and more complicated mathematical structures. The crowning achievement of these efforts are the theories of General Relativity and the (quantum) standard model of particle physics. The first one describes the behaviour of the universe at large scales and is based on the differential geometry developed by Riemann. The second one describes the fundamental particles of the universe and their interactions at short scales, using the algebraic techniques of gauge theories. An ultimate model, which is still out of sight, would combine aspects of both large and short scales, describing at the same time gravitation and the other interactions. In order to reach this goal, mathematicians and physicists need to associate with each other and propose novel physical models based on sophisticated mathematical structures.

The theory of noncommutative geometry, initially introduced by Connes \cite{lukbibl28}, is one of the new theories that provides an interesting mathematical background suitable for physical models, especially concerning the problem of unifying fundamental interactions. In this framework one uses a strong relation (equivalence of suitable categories) between geometry and algebra in order to translate the usual Riemannian geometry to the formalism of commutative $C^*$-algebras. Next, the geometric concepts are extended to general noncommutative $C^*$-algebras, hence the name of \emph{noncommutative geometry}. In this formalism one can formulate the \emph{noncommutative standard model} -- the classical standard model of particles coupled to Euclidean gravity \cite{Chamseddine:2006ep,ncg-bookC}.

While the noncommutative standard model already yields several interesting predictions in particle physics, the gravitational part of this model -- the part of noncommutative geometry which corresponds to the usual Riemannian geometry -- has mainly been developed using Euclidean signature, 
with no distinction between the time and spatial coordinates. As a consequence, the typical Lorentzian aspects of space-time coming from the differences between timelike and spacelike directions are lost in the current model. The development of \emph{Lorentzian noncommutative geometry} -- a Lorentzian counterpart of Connes' theory incorporating the specificities of Lorentzian signature -- is very recent and far from being complete. We shall present in this expository article the basic ingredients of Lorentzian noncommutative geometry along with our latest results. In particular, we will focus on how to define the causal structure of a, possibly noncommutative, space-time. We will show that in the almost commutative setting, causality plays an important role in the coupling between the classical space-time and the internal noncommutative space.

\section{Noncommutative geometry as a new tool for physical models}\label{secintrononcommu}

Geometry is probably the most useful mathematical concept in the construction of physical models. Typically, one uses at first a topological space (Hausdorff space, manifold), chooses an additional structure on it (differential structure, metric tensor, symplectic form, \ldots) and then tries to set some principles describing the dynamics (equations of motion, Lagrangian, Hamiltonian, \ldots). The construction of noncommutative geometry follows the same three steps, whose respective names could be:
\begin{itemize}\itemsep=3pt
\item Gelfand-Naimark duality;
\item Spectral triples;
\item Spectral action.
\end{itemize} \vspace{-0.5cm}

The powerfulness of noncommutative geometry consists in the fact that in the first step one enlarges the set of topological spaces to some `noncommutative spaces', which allow only a global, algebraic description. The next two steps are just an attempt to extend the usual structures and dynamical principles from classical geometric spaces to the new constructs. As a result, noncommutative geometry recovers the standard geometric models, at the same time extending the old concepts into an unexplored world, which contains new structures and new dynamics. In this section, we shall briefly present the key elements of these three steps.

\paragraph{The topology: Gelfand-Naimark theorem.}

If we consider a locally compact Hausdorff space $X$, then the space $C_0(X)$ of complex-valued continuous functions vanishing at infinity is a Banach algebra which is commutative and respects the norm condition $\norm{f \cdot f^*}_{\infty} = \norm{f}_{\infty}^2$. In general, we define a $C^*$-algebra as an involutive Banach algebra $A$ satisfying $\norm{a a^*} = \norm{a}^2$. Hence, $A=C_0(X)$ is a particular -- commutative -- case of a $C^*$-algebra, which is unital whenever the Hausdorff space $X$ is compact.

The famous Gelfand-Naimark Theorem guarantees that every commutative $C^*$-algebra emerges in this way.

\begin{theorem}
If $A$ is a commutative $C^*$-algebra, then the Gelfand transform \vspace{-0.4cm}
$$ \bigvee : A \rightarrow C_0(\Delta(A)) : a \leadsto \hat a \quad\text{defined by}\quad \hat a(\chi) = \chi(a) \vspace{-0.3cm}
$$ is an isometric *-isomorphism.
\end{theorem}

The set $\Delta(A)$, called the spectrum of $A$, is the space of $^*$-homomorphisms from $A$ to $\setC$ (usually called characters). $\Delta(A)$ can always be endowed with a locally compact (or compact if $A$ is unital) Hausdorff topology. This means that while any locally compact Hausdorff space gives rise to a commutative $C^*$-algebra, the converse is also true due to the Gelfand-Naimark theorem.

While dealing with a $C^*$-algebra $A$, we can define the set of states $S(A)$ -- positive linear functionals of norm one. The subset $P(A) \subset S(A)$ consists of pure states which are extremal points of $S(A)$ (i.e.~states that cannot be written as a convex combination of two other states). The set $S(A)$ is equipped with a natural weak-$^*$ topology. When the algebra is commutative, we have $P(A) \simeq \Delta(A)$ as topological spaces.

The topological space of our geometric model can now easily be defined. Let us take for example a manifold $\M$ (which is a locally compact Hausdorff space) and fix a commutative $C^*$-algebra $A = C_0(\M)$. To each point $x\in\M$ of the manifold there corresponds a pure state $\chi \in P(A)$ via the relation $\chi(a)=a(x)$, and every pure state can be written in this way. Hence, pure states can be considered as the \emph{points} of a topological space completely determined by the choice of the commutative $C^*$-algebra $A$. By taking this correspondence along to the realm of noncommutative $C^*$-algebras, we have the first setting of a \emph{noncommutative space}. As $P(A)$ comes with the natural weak-$^*$ topology, we actually can speak of a \emph{noncommutative topological space}.

\paragraph{The differential structure: Spectral triples.}

Given a compact Riemannian manifold with a spin structure $S$ one can endow the topological space determined by the $C^*$-algebra $A=C_0(\M)$ with a differential structure. First note that $\H=L^2(\M,S)$ -- the space of square integrable sections of the spinor bundle over $\M$ -- is a Hilbert space on which the elements of the algebra $A$ act as bounded operators. Next, with the help of the spin connection $\nabla^S$, one can define a Dirac 
 operator $D = -i(\hat c \circ \nabla^S) = -i e^\mu_a\gamma^{a} \nabla^S_\mu$, where $e^\mu_a$ stand for vielbeins and $\gamma^{a}$ are the flat gamma matrices. $D$ acts as an unbounded self-adjoint operator on $\H$ and has a compact resolvent. In order to extend such a structure to noncommutative $C^*$-algebra, the following axioms were proposed:

\begin{definition}
A spectral triple is given by the data $(\mathcal{A},\H,D)$
with:
\begin{itemize}\itemsep=3pt
\item A~Hilbert space $\H$.
\item A~unital pre-$C^*$-algebra $\mathcal{A}$, with a~faithful representation as bounded operators on
$\H$.
\item An unbounded essentially self-adjoint operator $D$ on $\H$, with a compact resolvent and such that $\forall a\in\mathcal{A}$, $[D,a]$ extends to a~bounded operator on $\H$.
\end{itemize}
\end{definition} 

In the commutative case, the condition that $[D,a]$ is a bounded operator restricts the algebra to the pre-$C^*$-algebra $\mathcal{A}=C^\infty(\M)$ of smooth functions (or at least to the algebra of Lipschitz continuous functions). The axioms can be adapted in order to deal with noncompact (but still complete) Riemannian manifolds (see \cite{IochumMoyal} for instance), by using the non-unital pre-$C^*$-algebra $\mathcal{A}=C_0^\infty(\M)$ and by replacing the compact resolvent condition by the requirement of  compactness of $a(1 + D^2)^{-\frac 12}$ $\forall a\in\mathcal{A}$.

An additional structure, like the parity or reality operator, can be set on a spectral triple.

\begin{definition}
A spectral triple is called even if there exists a $\mathbb Z_2$-grading $\gamma$ such that $[\gamma,a] = 0\ \forall a\in\mathcal{A}$ and $\gamma D = -D \gamma$.
\end{definition}

\begin{definition}
A spectral triple is called real of KO-dimension $n \in\mathbb Z_8$ if there exists an antilinear isometry $J :\mathcal{H} \rightarrow \mathcal{H}$  such that:
\begin{align*}
J^2 = \epsilon, && J D = \epsilon' D J, && J \gamma = \epsilon'' \gamma J, && [a,b^ \circ] = 0, && [[D,a],b^ \circ] = 0
\end{align*}
$\forall a,b\in\mathcal{A}$ with $b^ \circ =Jb^*J^{-1}$ and where the numbers $\epsilon, \epsilon', \epsilon'' \in \left\{{-1,1}\right\}$ depend on the value of $n\!\mod 8$:
\begin{center}
\begin{tabular}{|c|rrrrrrrr|}
 \hline
  n &0 &1 &2 &3 &4 &5 &6 &7 \\
\hline \hline
$\epsilon $  &1 & 1&-1&-1&-1&-1& 1&1 \\
$\epsilon'$ &1 &-1&1 &1 &1 &-1& 1&1 \\
$\epsilon''$&1 &{}&-1&{}&1 &{}&-1&{} \\  
\hline
\end{tabular}
\end{center}
\end{definition}

When the algebra is commutative and corresponds to a given Riemannian manifold, the complete information about the differential and metric structures can be recovered with the help of the Dirac operator $D$. Indeed, an element $[D,a]$ corresponds to a one form $da$ through the Clifford action $[D,a]=-ic(da)$ and the standard geodesic distance between two points $p,q\in\M$ is recovered by a purely algebraic formula:
\begin{equation} \label{eqnconnesdist}
d(p,q) = \sup\set{  \abs{a(q)-a(p)} \ :\  a \in \A,\  \norm{[D,a]} \leq 1 }.
\end{equation}
Such a formula admits a natural extension to noncommutative $C^*$-algebras by replacing the points $(p,q)$ with the corresponding pure states.

Hence, the definition of a spectral triple provides a differential and a metric structure for noncommutative manifolds. We can remark that, under some additional assumptions, every spectral triple the algebra of which is commutative corresponds to some compact spin Riemannian manifold on the strength of the famous Reconstruction Theorem \cite{C13}.

\paragraph{The dynamics: Spectral action.}

The dynamics in noncommutative geometry is governed by the \emph{spectral action principle} \cite{ConnesSA}. It assumes that the physical action only depends on the spectrum of the Dirac operator. Thus, the corresponding Lagrangian formalism is given by the following formula:
$$ S = \Tr\left( f \Big( \frac{D_A}{\Lambda} \Big) \right) $$
where $f$ is a positive even function, $\Lambda$ is a scale parameter and $D_A = D + A + \epsilon' J A J^*$ is the fluctuated Dirac operator with a Hermitian one-form $A= \sum_i a_i [D,b_i]$, $a_i,b_i\in\A$. Typically, the function $f$ is chosen to be a smooth approximation of a cut-off function. In such case, the spectral action simply counts the number of eigenvalues of the Dirac operator $D_A$ smaller than the scale $\Lambda$.

When the spectral triple is built from a 4-dimensional compact Riemannian manifold with the standard Dirac operator $D = -i e^\mu_a\gamma^{a} \nabla^S_\mu$, the spectral action yields the following action:
$$ S = \Tr\left( f \Big( \frac{D}{\Lambda} \Big) \right) = \int_\M \mathcal{L_M}   \sqrt{\abs{g}} d^4x +  \mathcal{O}(\Lambda^{-1})$$
and the computation of the Lagrangian density gives \cite[Proposition 3.3]{Dungen}
$$ \mathcal{L_M}   = \frac{f_4}{2\pi^2}\Lambda^4 + \frac{f_2}{24 \pi^2}R \Lambda^2 + \frac{f(0)}{16 \pi^2} \left(  -\frac{1}{30} \Delta R - \frac{1}{20} C_{\mu\nu\rho\sigma} C^{{\mu\nu\rho\sigma}} + \frac{11}{360} R^* R^* \right).$$
The function $f$ enters the formula only through its moments \linebreak $f_i \! = \! \int_0^{\infty} \! f(x) x^{i-1} dx$. The first term gives information about the volume, the second one provides the scalar curvature and the third one contains the Weyl tensor and a Gauss-Bonnet topological invariant.

We have seen that in noncommutative geometry one is able to recover the usual geometry (at least with Euclidean signature). Moreover, the construction of new noncommutative $C^*$- algebras offers the possibility to extend geometric concepts to noncommutative spaces and to unify them with commutative ones. A standard way is to construct an \emph{almost commutative manifold} which is a kind of Kaluza-Kein product between a continuous space and a discrete space. The continuous space $(\mathcal{A}_\mathcal{M},\H_\mathcal{M},D_\mathcal{M},(\gamma_\mathcal{M}))$ is constructed from a Riemannian spin manifold (for simplicity with even dimension) while the finite space $(\mathcal{A}_F,\H_F,D_F)$ is built from a noncommutative algebra $\mathcal{A}_F$. Then, the following product is a well defined spectral triple: \vspace{-0.5cm}
\begin{eqnarray}
&\mathcal{A}=\mathcal{A}_\mathcal{M}\otimes\mathcal{A}_F,\nonumber\\
&H=H_\mathcal{M}\otimes H_F, \label{eqnacproduct}\\
&D=D_\mathcal{M}\otimes1+\gamma_\mathcal{M}\otimes D_F.\nonumber
\end{eqnarray}
With a smart choice of the finite algebra, the unitary group of which corresponds to the gauge group of the Standard Model of elementary particles, one obtains a robust model of fundamental interactions \cite{Chamseddine:2006ep}. A complete review on the subject can be found in \cite{Dungen}.

Usually, the finite algebra is taken to be $\mathcal{A}_F = \mathbb{C} \oplus \mathbb{H} \oplus M_3(\mathbb{C})$. However, one could equivalently work with $\widetilde{\mathcal{A}}_F = \mathbb{C} \oplus M_2(\mathbb{C}) \oplus M_3(\mathbb{C}) \oplus \mathbb{C}$ \cite{Schucker,Christoph}. The latter option is seldom mentioned, however it seems to be better suited from the point of view of the space of states. We will elaborate on this in the concluding section.

\section{Lorentzian structures in noncommutative geometry}\label{secLostruct}

Since noncommutative geometry provides a wonderful tool for the construction of new physical models, it would be logical to apply the formalism to Lorentzian rather than Riemannian manifolds to encompass properly the theories of gravity. However, whereas in the Riemannian (and most of the time compact) setting noncommutative geometry is founded on solid mathematical grounds, in the Lorentzian case the situation is much more cumbersome. The key element of the theory -- the spectral properties of the Dirac operator -- are much harder to describe when the metric is not positive definite. Typically, on a Lorentzian manifold $D$ is not self-adjoint, it has an infinite dimensional kernel and therefore it cannot have a compact resolvent \cite{Baum}. For these reasons, the Lorentzian version of noncommutative geometry is still very much in the shadows, despite several courageous attempts \cite{Mor,Pas,Stro,F4,F6,Rennie12}. In this section, we shall discuss the basic setting for Lorentzian noncommutative geometry using the knowledge available at the present time. In the next sections, we will present some intriguing recent results obtained by the authors concerning the notion of causality in noncommutative geometry.

Since the signature of the metric does not affect the topology of the space, the construction of the $C^*$-algebra and the \emph{space of states} remains identical to the standard one presented in the previous section. One only has to be aware that in the Lorentzian setting compact manifolds are always acausal as they contain closed causal curves. Since one of our goals is to describe the causal structure, we will always work in the locally compact framework. Hence, the algebra carrying the topological information has to be non-unital -- for instance in the commutative case one chooses a suitable subalgebra of $C^\infty_0(\M)$. However, since we will be concerned with causal functions, which are non decreasing along causal curves, we will need a bigger algebra $\widetilde\A$ -- a preferred unitisation of $C^\infty_0(\M)$. Typically it is a specific subalgebra of $C^\infty_b(\M)$ -- the space of smooth bounded functions corresponding to a particular compactification of the manifold.

In Lorentzian signature, we can still consider the Dirac operator $D = -i(\hat{c}\circ\nabla^S) = -i e^\mu_a\gamma^{a} \nabla^S_\mu$ associated with a spin structure. Now, the (flat) gamma matrices respect $\gamma^{a} \gamma^{b} + \gamma^{b} \gamma^{a} = 2 \eta^{ab}$.\footnote{Our convention is such that $\gamma^0$ is anti-Hermitian, $\gamma^a$ are Hermitian for $a>0$ with the signature of the metric being $(-,+,+,+,\dots)$.} However, the space $\Gamma(M,S)$ of sections of the spinor bundle over $M$ does not have a natural Hilbert space structure. On the other hand, there exists a \emph{natural} bilinear form $(\cdot,\cdot)_x$ which is of split signature and endows $\Gamma(M,S)$ with the non-degenerate inner product:
$$ (f,g) = \int_\M (f_x,g_x)_x \sqrt{\abs{g}} d^nx. $$
But since this inner product is not positive definite, the structure is not the one of a Hilbert space but of a Krein space \cite{Bog}. However, there exists a second possibility, which gives a well defined Hilbert space structure. Instead of $(\cdot,\cdot)_x$ one could use a positive definite bilinear form $\scal{ \cdot,\cdot}_x$ over $S$ (as e.g.~the natural scalar product in $\mathbb{C}^{2^{\lfloor n/2\rfloor}}$) and define the positive definite inner product:
$$ \scal{f,g} = \int_\M \scal{f_x,g_x}_x \sqrt{\abs{g}} d^nx. $$
The Hilbert space defined as the completion under this positive definite inner product is denoted by $\H=L^2(\M,S)$. The relation between $(\cdot,\cdot)_x$ and $\scal{ \cdot,\cdot}_x$ is provided with the help of Clifford multiplication by a normalised timelike vector field $e$ as $\scal{ \cdot,\cdot}_x=(\cdot, i c(e_x)\cdot)_x$. This gives rise to a bounded self-adjoint operator $\J = i c(e)$ on $\H$ respecting $\J^2=1$ and intertwining between the two scalar products:
\begin{align*}
(\cdot,\cdot) = \scal{\cdot,\J \cdot} && \text{and} && \scal{\cdot,\J \cdot} =  (\cdot,\cdot).
\end{align*}
Such an operator is called a \emph{fundamental symmetry} and turns the Hilbert space into a Krein space and vice versa. A fundamental symmetry is somewhat similar to a Wick rotation and can often be associated with a spacelike reflexion, i.e.~a linear map on the tangent bundle respecting $r^2=1$, $g(r\cdot,r\cdot) = g(\cdot,\cdot)$ and such that $g^r(\cdot,\cdot)=g(\cdot,r\cdot)$ is a Riemannian metric.

Now that we have a Hilbert space, we can consider the action of $\A$ and $D$ on it. However, whereas in the Riemannian case the Dirac operator was (essentially) self-adjoint on $\H$, this is no longer the case in the Lorentzian setting. Instead, the self-adjointness property must be reformulated using the Krein space structure.

\begin{theorem}[\cite{Baum,Stro}]\label{thstro}
If there exists a spacelike reflection $r$ such that the Riemannian metric $g^r$ associated with this reflection is complete, then $i D$ is essentially Krein-selfadjoint, i.e.~essentially self-adjoint for the indefinite inner product $ (\cdot,\cdot)$.\\
If $\J$ is the fundamental symmetry associated with $r$, then $i \J D$ extends to a self-adjoint operator on $\H$ or equivalently $\D^* = -\J D \J$.
\end{theorem}

In the following, we will refer to \emph{complete} Lorentzian manifolds in the sense of the completeness condition described in Theorem \ref{thstro}.

We thus have a prescription to construct a triple $(\A,\H,D)$ from a complete Lorentzian manifold with an additional operator $\J$. Now, we shall try to axiomatise the construction in a completely algebraic setting and this is where a new open problem emerges. The freedom of the choice of the fundamental symmetry $\J$ implies in general that the signature of the manifold is only pseudo-Riemannian \cite{Pas,Rennie12}. One way to guarantee the correct Lorentzian signature of the metric is to force the fundamental symmetry to correspond to the first flat gamma matrix $\J=i \gamma^0$. A possibility to impose this condition is to fix the fundamental symmetry as a \emph{simple one-form} $\J = - N [D,\T]$. However, such a one-form can be constructed in an unambiguous way only on globally hyperbolic manifolds \cite{F4,F5}, therefore we now restrict ourselves to this class.

\begin{proposition}
If $(\M,g)$ is a complete globally hyperbolic manifold, then $\J = - N [D,\T]$ is a suitable fundamental symmetry where $\T$ is a smooth global time function (temporal function) which splits the metric and $N$ is the lapse function.
\end{proposition}

\begin{proof}
On the strength of the smooth version of Geroch's splitting theorem \cite{BS04}, there always exists a global smooth function $\T$ on $\M$ such that the metric reads $g = - N^2 d^2\T + g_\T$ where $g_\T$ is a set of Riemannian metrics defined on the Cauchy hypersurfaces at constant $\T$. We note that the completeness condition of Theorem \ref{thstro} simply says that the manifold $\M$ needs to be complete under the metric $g^r = N^2 d^2\T + g_\T$. Here $d\T$ is a timelike vector field and, up to a normalisation, it defines a fundamental symmetry $\J=\Omega i c(d\T) = \Omega i \tilde\gamma^0$, where $\tilde\gamma^0$ denotes the first \emph{curved} gamma matrix respecting $(\tilde\gamma^0)^2=g^{00}=-\frac{1}{N^2}$ since the metric splits. We have $\tilde\gamma^0 = N^{-1}\gamma^0$ which implies $\J^2 = \Omega^2 N^{-2} = 1$, so the normalisation is given by the lapse function. From the relation $[D,\T]=-ic(d\T)$, we get the formula $\J=-N [D,\T]$.
\end{proof}

\vspace{0.2cm}

Let us note that typically the time function $\T$ is unbounded, since globally hyperbolic space-times are not compact. Therefore, $[D,\T]$ need not in general be bounded either. However, $\T$ (as strictly growing along future directed curves) can be rescaled to a bounded function, for instance via a composition with $\arctan$. Hence, as $\T$ is smooth, one can always choose the time function $\T$ in such a way that $[D,\T]=-ic(d\T)$ is bounded.

Also, it is convenient to assume that the lapse function, which is a smooth bounded function, is a Hermitian element in the preferred unitisation $\widetilde\A \subset C^\infty_b(\M)$. In order to guarantee a suitable  representation of the time function $\T$ on $\H$, we will require that $\T$ is a self-adjoint operator on $\H$ with a suitable domain and such that $(1+\T^2)^{-\frac{1}{2}} \in \widetilde\A$. We now show that such a choice of fundamental symmetry restricts the possible signatures to those of a Lorentzian form.

\vspace{0.2cm}

\begin{proposition}
Let us assume that a Lorentzian triple $(\A,\H,D)$ with a fundamental symmetry $\J$ corresponds to a complete pseudo-Riemannian spin manifold $(\M,g)$. If the fundamental symmetry is given by $\J=-N  [D,\T]$ for some smooth functions $N$ and $\T$ on $\M$, then the geometry is Lorentzian and the metric admits a global splitting.
\end{proposition}

\vspace{0.2cm}

\begin{proof}
For sake of clarity, we will denote by $\tilde\gamma^\mu = e^\mu_a \gamma^a$ the curved gamma matrices. The Dirac operator reads $D = -i \tilde\gamma^{\mu} \nabla^S_\mu$ and the Clifford relations are $\tilde\gamma^{\mu} \tilde\gamma^{\nu} + \tilde\gamma^{\nu} \tilde\gamma^{\mu} = 2 g^{\mu\nu}$, where $\tilde\gamma^{\mu}$ are chosen to be either Hermitian or anti-Hermitian. In arbitrary local coordinates, the fundamental symmetry reads $\J=N i \tilde\gamma^\mu \partial_\mu\T$ so the condition $\J^2=1$ can only be respected if the gradient of $\T$ does not vanish anywhere. Hence, $\T$ can be chosen as a first (global) coordinate $x^0=\T$ and we get $\J=N i \tilde\gamma^0$. Then $\J^2 = - N^{2} g^{00} = 1 \implies g^{00} = -N^{-2} < 0 $. We also note that $\tilde\gamma^0$ is forced to be anti-Hermitian due to the self-adjointness of $\J$.

From the relation $D^*=-\J D\J \Longleftrightarrow (\J D)^*  + \J D = 0$ and following the arguments of \cite[Proposition 9.13]{Gracia-Bondia} (see also \cite{Baum}), we get
$$  
(\tilde\gamma^\mu)^* i \tilde\gamma^0 (-i \nabla^S_\mu)^*+ i \tilde\gamma^0 \tilde\gamma^\mu  (-i \nabla^S_\mu) = \left( (\tilde\gamma^\mu)^*  \tilde\gamma^0 +  \tilde\gamma^0 \tilde\gamma^\mu \right) \nabla^S_\mu = 0. 
$$
Now we can discuss the Hermicity of the matrices $\tilde\gamma^\mu$ for $\mu>0$. If $\tilde\gamma^\mu$ has the same Hermicity as $\tilde\gamma^0$ (so is anti-Hermitian), we get $\tilde\gamma^\mu  \tilde\gamma^0 =  \tilde\gamma^0 \tilde\gamma^\mu$ which implies $ \tilde\gamma^\mu \tilde\gamma^0 = g^{\mu 0} $ by the Clifford relations. But this is impossible since  $\tilde\gamma^0 \tilde\gamma^\mu$ cannot be a multiple of the identity matrix for $\mu \neq 0$. Hence, all $\tilde\gamma^\mu$ are Hermitian for $\mu > 0$ and correspond to positive eigenvalues of the metric. Moreover, we get $\tilde\gamma^\mu  \tilde\gamma^0 +  \tilde\gamma^0 \tilde\gamma^\mu = 0 = g^{\mu0} = g^{0\mu}$ and the metric admits a global splitting.
\end{proof}

Let us stress, however, that the above proposition cannot be promoted to a Lorentzian version of the reconstruction theorem for spectral triples in a straightforward way. Rather, it should be understood as follows: assuming that there exists a reconstruction theorem for pseudo-Riemannian spin manifolds, the proposed form of the fundamental symmetry restricts the possibilities to manifolds with a Lorentzian signature. A complete reconstruction theorem analogous to that of \cite{C13} is beyond our reach for the moment. Also, the chosen fundamental symmetries do not cover all possible Lorentzian spin manifolds but only those which admit a global splitting. In particular, these encompass the set of complete globally hyperbolic manifolds.

We now come to our \emph{practical} definition of a Lorentzian spectral triple. This definition is certainly not an ultimate one and is supposed to evolve in the course of future investigation. On the other hand, we take it as a base for our studies of causality as it covers a large part of Lorentzian manifolds, in particular the globally hyperbolic space-times.

\begin{definition}
\label{deflost}
A~Lorentzian spectral triple is given by the data $(\mathcal{A},\widetilde{\mathcal{A}},\H,D,\mathcal{J})$
with:
\begin{itemize}
\itemsep=0pt \item A~Hilbert space $\H$.
\item A~non-unital pre-$C^*$-algebra $\mathcal{A}$, with a~faithful representation as bounded operators on
$\H$.
\item A~preferred unitisation $\widetilde{\mathcal{A}}$ of $\mathcal{A}$, which is also a~pre-$C^*$-algebra
with a~faithful representation as bounded operators on $\H$ and such that $\mathcal{A}$ is an ideal of
$\widetilde{\mathcal{A}}$.
\item An unbounded operator $D$, densely defined on $\H$, such that:
\begin{itemize}\itemsep=0pt
\item $\forall a\in\widetilde{\mathcal{A}}$,   $[D,a]$ extends to a~bounded operator on $\H$, \item $\forall
a\in\mathcal{A}$,   $a(1 + \scal{D}^2)^{-\frac 12}$ is compact, with $\scal{D}^2 = \frac 12 (D D^* + D^*
D)$.
\end{itemize} 
\item A~bounded operator $\mathcal{J}$ on $\H$ such that:
\begin{itemize}\itemsep=0pt
 \item $\mathcal{J}^2=1$, \item $\mathcal{J}^*=\mathcal{J}$, \item $[\mathcal{J},a]=0$,
$\forall a\in\widetilde{\mathcal{A}}$, \item $D^*=-\mathcal{J} D \mathcal{J}$, \item
$\mathcal{J}  = -N [D,\mathcal{T}]$ for $N\in\widetilde\A$, $N>0$ and for some (possibly unbounded) self-adjoint operator $\mathcal{T}$ with domain
\mbox{$\text{Dom}(\mathcal{T})  \subset  \H$} and such that  $\left(1+ \mathcal{T}^2 \right)^{-\frac{1}{2}}\in
\widetilde{\mathcal{A}}$.
\end{itemize}
\end{itemize}
\end{definition}

\begin{definition}
A~Lorentzian spectral triple is even if there exists a~$\mathbb Z_2$-grading $\gamma$ such that
$\gamma^*=\gamma$, $\gamma^2=1$, $[\gamma,a] = 0$ $\forall a\in\widetilde{\mathcal{A}}$, $\gamma
\mathcal{J} =- \mathcal{J} \gamma$ and $\gamma D =- D \gamma $.
\end{definition}

\section{Causality in noncommutative geometry}

Causality is a key concept in natural sciences and especially in physics (see \cite{CausalUniverse} for a comprehensive review). It puts fundamental restrictions on the admissible physical processes independently of any concrete model. Causality is the cornerstone of Special Relativity (and a fortiori its general version), hence it may be viewed as the most important information coming from the Lorentzian signature of the space-time. Also, it is one of the basic axioms of every quantum field theory. Therefore, when considering the physical models based on noncommutative geometry it is of utmost importance to understand their causal structure and its implications for physics.

Causality can be given a geometrical interpretation as a link between two points (events), i.e.~specific positions at specific times. On a (path) connected smooth manifold, every two points can be connected by a curve, so, a priori, they can be linked. If we consider a Lorentzian manifold, we can distinguish specific  types of curves in function of the behaviour of their tangent vectors. In particular, we discern timelike curves $\gamma$, the tangent vector of which is everywhere timelike ($g(\dot \gamma,\dot \gamma) < 0$) and causal curves, where the tangent vector is allowed also to be null ($g(\dot \gamma,\dot \gamma) \leq 0$). Along such curves a physical signal can be transmitted  with a speed not exceeding the speed of light (in vacuum).  Two points (events) are then said to be  \emph{causally related} if there exists a causal curve linking them. Moreover, if the manifold is time-oriented, specific future and past directions can be chosen. This leads to the familiar cone picture.

The nature of causality in mathematical physics is therefore entirely geometric and the information is extracted using curves and tangent vectors. Since there is no notion of curves or tangent vectors in noncommutative geometry, the point of view on causality must be completely revised, and the description should only depend on the data included in Definition \ref{deflost}.

The first question one is confronted with when leaving the peaceful land of commutative geometry, is: `Where do the causal relations actually take place?'. In other words: `What are the ``noncommutative'' events?'. We have already argued that pure states of a (possibly noncommutative) $C^*$-algebra could be seen as a counterpart of points. On the strength of this correspondence, we shall regard the space of states as a proper setting for the causal structure. Let us stress however, that there are other possibilities to define points of a noncommutative space -- for instance using maximal ideals of irreducible representations. These notions are all equivalent in the commutative case, but bifurcate in the noncommutative one (see the chapter by A. Sitarz in this volume). In the last section, where some physical consequences of our results will be discussed, we will elaborate on the motivation behind using the space of states.

Next, we need to characterise the causal curves in an algebraic manner. In accord with the philosophy of noncommutative geometry, this should be done in a dual way by using a specific class of functions called \emph{causal functions}.

\begin{definition}
A causal function is a real-valued function which is non-decrea\-sing along every future-directed causal curve.
\end{definition}

Global time functions and constant functions are examples of causal functions. In order to comply with the formalism of Lorentzian spectral triples, we will restrict to a subclass of smooth bounded causal functions belonging to $\widetilde\A$. The set of such causal functions is denoted by $\C \subset \widetilde\A$ and has the structure of a convex cone, i.e.~$\forall f,g \in \C$, $\forall \lambda \geq 0$, $\lambda f + g \in \C$. The set $\C$ spans the whole algebra, i.e.~$\Span_{\setC}(\C)$ is a unital pre-$C^*$-algebra which contains $\A$. Since the choice of the preferred unitisation $\widetilde\A$ is free, we will always choose $\widetilde\A=\Span_{\setC}(\C)$. Then, the following proposition holds:

\begin{proposition}[\cite{Bes}]
If $\M$ is a globally hyperbolic space-time, then the set of causal functions completely determines the causal structure of $\M$ by:
$$\forall p,q\in\M,\qquad p \preceq q \quad \text{ iff }\quad \forall f\in\C,\ f(p) \leq f(q).$$
\end{proposition}

Therefore, we have a complete characterisation of causality, which uses the cone $\C$ of causal functions instead of curves.

The most difficult part however, is to have an algebraic characterisation of the cone $\C$. Indeed, if one takes an arbitrary convex cone (with a few additional axioms), then it is possible to recover every possible partial order structure on the manifold \cite{Bes}. In order to restrict the possible partial order structures to those corresponding to a causal structure defined by a Lorentzian metric, we shall introduce an operatorial constraint based on the Dirac operator $D$ and the fundamental symmetry $\J$.

\begin{theorem}[\cite{F6}]\label{thmopcausal}
Let $(\A,\widetilde\A,\H,D,\J)$ be a commutative Lorentzian spectral triple constructed from a complete globally hyperbolic space-time $\M$. Then $f\in\widetilde\A$ is a causal function if and only if 
$$\forall \phi \in \H, \quad\scal{\phi,\J[D,f] \phi} \leq 0,$$
where $\scal{\cdot,\cdot}$ is the inner product on $\H$.
\end{theorem}

This theorem can be explained intuitively. We saw in Section \ref{secLostruct} that with a smooth global time function $\T$, the operator $\mathcal{J}  = -N [D,\mathcal{T}]$ allows us to go back and forth from the positive inner product of the Hilbert space $\H$ to the corresponding Krein inner product (so we have $(\phi,\J \phi) = \scal{\phi,\J^2 \phi} = \scal{\phi,\J \left(-N[D,\mathcal{T}] \right ) \phi} > 0$ for $\phi \neq 0$). If instead of using the non-degenerate operator $[D,\mathcal{T}]$ we extend to a possibly degenerate operator $[D,f]$ with $f\in\C$ (a function non-decreasing along future directed causal curves instead of strictly increasing), then we recover a non-negative degenerate inner product $\scal{\phi,\J \left(-N[D,f] \right ) \phi} \geq 0$. The requirement of such a condition fulfilled for every $\phi\in\mathcal{H}$ can be taken as a complete characterisation of a causal function. However, this is only an intuitive explanation. The complete formal proof of this theorem, including necessary and sufficient conditions, can be found in \cite{F6}.

The operatorial formulation of Theorem \ref{thmopcausal} encourages us to propose the following  definition of a causal structure for noncommutative geometries.

\begin{definition}\label{defncauscone}
Let us consider the cone $\mathcal{C}$ of all Hermitian elements $a \in \widetilde{\mathcal{A}}$ respecting
\begin{equation}\label{eqopcauscond}
\forall \; \phi\in\H,\qquad\scal{\phi,\mathcal{J}[D,a]\phi}\leq0,
\end{equation}
where $\scal{\cdot,\cdot}$ is the inner product on $\H$.
If the following condition is fulfilled:
\begin{equation} \label{eqncondcausal}
\overline{\Span_{{\mathbb C}}(\mathcal{C})}=\overline{\widetilde{\mathcal{A}}},
\end{equation}
then $\mathcal{C}$ is called a \emph{causal cone} and defines a partial order relation on
$S(\widetilde{\mathcal{A}})$ by
$$\forall \omega,\eta\in S(\widetilde{\mathcal{A}}),\qquad\omega\preceq\eta\qquad\text{iff}\qquad\forall a\in\mathcal{C},\quad\omega(a)\leq\eta(a).$$
\end{definition}

The above definition is further motivated by the following theorem:

\begin{theorem}[\cite{F6}]\label{thmreconstruction}
Let $(\mathcal{A},\widetilde{\mathcal{A}},\H,D,\mathcal{J})$ be a commutative Lorentzian spectral
triple constructed from a complete globally hyperbolic Lorentzian manifold $\mathcal{M}$, and let us define the
following subset of pure states:
$$\mathcal{M}(\widetilde{\mathcal{A}})=\big\{\omega\in P(\widetilde{\mathcal{A}}) : \mathcal{A}\not\subset\ker\omega\big\}\subset S(\widetilde{\mathcal{A}}).$$
Then the causal relation $\preceq$ on $S(\widetilde{\mathcal{A}})$ restricted to $\mathcal{M}(\widetilde{\mathcal{A}}) \cong P(\A) \cong \mathcal{M}$ corresponds to the usual causal relation on
$\mathcal{M}$.
\end{theorem}

Definition \ref{defncauscone} and Theorem \ref{thmreconstruction} clarify the role of the unitisation $\widetilde\A$ in the context of causality. Indeed, if the non-unital algebra $\A$ is used instead, then the only element respecting the operatorial condition \eqref{eqopcauscond} is the null element. Hence, the unitisation is technically needed in order to define a valid causal cone $\C$. However, the space of pure states $P(\widetilde\A)$ is larger that the space $P(\A)$ corresponding to the manifold, so we need to get rid of the \emph{extra points} by introducing the subspace $\mathcal{M}(\widetilde{\mathcal{A}})$. In the commutative case, the subspace $\mathcal{M}(\widetilde{\mathcal{A}})$ corresponds to the initial manifold, so the choice of the unitisation does not play a role. It is also the case for almost commutative algebras, as we will see in Section \ref{secalmcom}. However, for some \emph{truly} noncommutative spaces, different unitisations might lead to different causal structures.

The condition \eqref{eqncondcausal} is crucial, since it guarantees that the relation $\preceq$ defines a good partial order structure on the space of states. On a globally hyperbolic space-time, there always exists an unitisation $\widetilde\A$ such that this condition is respected \cite{Bes,F6}. However, if the manifold is acausal (as for example a compact Lorentzian manifold without boundary), such a condition cannot be fulfilled. Indeed, on such manifolds there always exists a closed timelike curve, on which every causal function must be constant. Hence the set of causal functions cannot span the algebra. Therefore, it is natural to consider this condition as a criterion to discern whether a space is causal or not also in the noncommutative regime.

As a digression let us note that although the definition of causality is a large step towards the understanding of Lorentzian metric aspect of noncommutative geometry, it is not the end of the story. Indeed, it is well known that a causal structure gives information about the metric only up to a conformal factor. 
Therefore, an algebraic formulation of a Lorentzian distance, i.e.~a Lorentzian counterpart to \eqref{eqnconnesdist}, is still an open question despite some proposals. In \cite{Mor}, a formula based on d'Alembert operator is presented, but its noncommutative generalisation is tedious and not applicable in the usual framework of spectral triples. In \cite{F6}, following a result of \cite{F3}, the authors have suggested the following practical formulation of a Lorentzian distance:
\begin{definition}\label{defnewdist}
Let $(\A,\widetilde\A,\H,D,\J)$ be an even commutative Lorentzian spectral triple with $\setZ_2$-grading $\gamma$ constructed from a complete globally hyperbolic space-time $\M$ of even dimension. For every two points $p,q\in\M$, we define:
\begin{align}
& \widetilde d(p,q) \label{LorDist} \\
& \;\; = \inf_{f \in C^\infty(\M,\setR)}\set{ \max\set{0,f(q)-f(p)} :
\forall \phi \in \H, \scal{\phi,\J([D,f]+ i\gamma) \phi } \leq 0}. \notag
\end{align}
\end{definition}

This formula (with points traded for states) seems to be a good candidate for a generalisation of the Lorentzian distance formula in noncommutative geometry, as it respects the following properties:
\begin{proposition}[\cite{F6}]
The function $\widetilde d(p,q)$ from Definition \ref{defnewdist} meets all of the properties of a Lorentzian distance, i.e.~
\begin{enumerate}
\item $d(p,p) = 0$,
\item $d(p,q) \geq 0$ for all $p,q\in\M$,
\item if $d(p,q) > 0$, then $d(q,p) = 0$,
\item if $d(p,q) > 0$ and $d(q,r) > 0$, then $d(p,r) \geq d(p,q) + d(q,r)$.\begin{flushright} (`wrong way' triangle inequality)\end{flushright}
\end{enumerate}
Moreover, for every $p,q\in \M$, we have $d(p,q) \leq \widetilde d(p,q)$ where $d(p,q)$ denotes the usual Lorentzian distance between $p$ and $q$.\\[0.3cm]
If $\M$ is the Minkowski space-time of even dimension, we have the equality $d(p,q) = \widetilde d(p,q)$.
\end{proposition}

\vspace{0.5cm}

As we see, formula \eqref{LorDist} yields a valid Lorentzian distance corresponding to the usual one on Minkowski space-time (and in fact on every space-time whose Lorentzian distance function can be approximated by smooth functions). Despite the fact that no counterexamples are known for more general space-times, the complete proof of the equality of \eqref{LorDist} with the usual Lorentzian distance is still an open problem. Also, a formulation in odd number of space-time dimensions has not yet been established.

The generalisation of \eqref{LorDist} to noncommutative space-times is not straightforward as the formula involves unbounded functions and might require a supplementary structure in addition to those included in the definition of a Lorentzian spectral triple. A proposal for such a structure suitable for the construction of a noncommutative Lorentzian distance formula is presented in \cite{F5}.

\vspace{0.2cm}

\section{An almost commutative toy model}\label{secalmcom}

As an illustration of the previous section we shall present a toy model of an almost commutative manifold on which the notion of causality can easily be explored. The technical details can be found in \cite{F7}. This model is constructed as a prelude to the study of the noncommutative standard model described in Section \ref{secintrononcommu}.

The toy model is constructed as a Kaluza-Klein product between a two-dimensional Minkowski space-time and a finite noncommutative space based on the complex algebra $M_2(\setC)$. The continuous Minkowski space-time will be referred to as the \emph{space-time} and the finite space as the \emph{internal space}.

\newpage The construction of space-time as a Lorentzian spectral triple is obtained as follows:
\begin{itemize}\itemsep=0pt
\item $\H_\mathcal{M} = L^2({\mathbb R}^{1,1}) \otimes {\mathbb C}^{2}$ is the Hilbert space
of square integrable sections of the spinor bundle over the two-dimensional Minkowski space-time.
\item $\mathcal{A}_\mathcal{M} = \mathcal{S}({\mathbb R}^{1,1})$ is the algebra of Schwartz functions
(rapidly decreasing at inf\/inity together with all derivatives) with pointwise multiplication.
\item $\widetilde{\mathcal{A}}_\mathcal{M} = \Span_{{\mathbb C}}(\mathcal{C}_\mathcal{M}) \subset
\mathcal{B}({\mathbb R}^{1,1})$ is the algebra of smooth bounded functions with all derivatives bounded, the limits of which exist along every causal curve. $\mathcal{C}_\mathcal{M}$ represents the set of smooth bounded causal functions on ${\mathbb R}^{1,1}$ with all derivatives bounded.
\item $D_\mathcal{M} = -i \gamma^\mu\partial_\mu$ is the f\/lat Dirac operator.
\item $\mathcal{J}_\mathcal{M}=-[D,x^0] = ic(dx^0) = i\gamma^0$, where $x^0$ is the global time coordinate.
\item $\gamma_\mathcal{M} = \gamma^0 \gamma^1.$
\end{itemize}
The internal space is a finite Riemannian spectral triple:
\begin{itemize}\itemsep=0pt
\item $\H_F = {\mathbb C}^2$; 
\item $\mathcal{A}_F = M_2({\mathbb C})$ is the algebra
of $2\times2$ complex matrices with the natural multiplication and representation on $\H_F$;
\item $D_F =
\diag(d_1,d_2)$, with $d_1,d_2 \in {\mathbb R}$, $d_1 \neq d_2$.
\end{itemize}
The Dirac operator in the internal space can be chosen to be diagonal, since any other self-adjoint operator can be obtained by a unitary transformation. The latter can be understood as a global rotation of the space of states and thus does not affect the causal structure.

The almost commutative spectral triple is constructed using the product \eqref{eqnacproduct} with the fundamental symmetry given by $\J=\J_\M \otimes 1$. The unitisation is chosen as the algebra of $2 \times 2$ complex matrices on $\M$ whose diagonal entries are in $\widetilde{\mathcal{A}}_\mathcal{M}$ and the off-diagonal ones -- in $\mathcal{A}_\mathcal{M}$, i.e.
$$\widetilde{\mathcal{A}}\cong\left(\begin{matrix}\widetilde{\mathcal{A}}_\mathcal{M} &  \mathcal{A}_\mathcal{M}\\
\mathcal{A}_\mathcal{M} &  \widetilde{\mathcal{A}}_\mathcal{M}\end{matrix}\right)=\left(
\begin{matrix}\Span_{{\mathbb C}}(\mathcal{C}_\mathcal{M}) &  \mathcal{S}\big({\mathbb R}^{1,1}\big)\\
\mathcal{S}\big({\mathbb R}^{1,1}\big) &  \Span_{{\mathbb C}}(\mathcal{C}_\mathcal{M})\end{matrix}\right) \cdot$$
This specific choice is dictated by the condition \eqref{eqncondcausal} along with the requirements of Definition \ref{deflost}. As in the commutative case, the space of `physical' pure states $\mathcal{M}(\widetilde{\mathcal{A}}) =\big\{\omega\in P(\widetilde{\mathcal{A}}) : \mathcal{A}\not\subset\ker\omega\big\} \cong \M \times P(\A_F)$ is independent of the unitisation.

At first we need to describe the space $\mathcal{M}$. Since the algebra $\mathcal{A}_\mathcal{M}$ is commutative, all of the pure states are \emph{separable} and we have $\mathcal{M}(\widetilde{\mathcal{A}}) \cong \M \times P(\A_F)$. Moreover, it is known \cite{IochKM} that $P(\M_2(\setC)) \cong {\mathbb C} P^1 \cong S^2$, so the space of pure states (with the \mbox{weak-$^*$} topology) is homeomorphic to the Cartesian product $\M \times S^2$. Let us stress however that the $S^2$ component is a only topological sphere and its metric aspect is governed by the noncommutative geometry described by the spectral triple $(\mathcal{A}_F,\H_F,D_F)$.

Every pure state $\omega_{p,\xi} \in \mathcal{M}(\widetilde{\mathcal{A}})$ is determined by two entries:
\begin{itemize}\itemsep=0pt
\item a~point $p \in {\mathbb R}^{1,1}$,
\item a~normalised complex vector $\xi\in{\mathbb C}^2$ defined up to a phase.
\end{itemize}
The value on some element $\mathbf{a}\in\widetilde{\mathcal{A}}$ is given by $\omega_{p,\xi}(\mathbf{a})=\xi^*\mathbf{a}(p)\xi$, where $\mathbf{a}(p)$ is the application of the evaluation map at $p$ on each entry in $\mathbf{a}$.

A vector $\xi=(\xi_1,\xi_2) \in {\mathbb C} P^1 \subset {\mathbb C}^2$ is identified with the point $(x_\xi,y_\xi,z_\xi) \in S^2$ via the relations $x_\xi=2\Re\{\xi_1^*\xi_2\}$, $y_\xi=2\Im\{\xi_1^*\xi_2\}$ and $z_\xi=\abs{\xi_1}^2-\abs{\xi_2}^2$. The quantity $z_\xi$ is named the latitude of the vector $\xi$ on $S^2$ and plays a particular role in the description of the causal structure. We note that if two pure states $\xi$ and $\varphi$ have the same latitude, they can always be written as $\xi_2 = \abs{\xi_2}
e^{i\theta_\xi}$ and $\varphi_2 = \abs{\varphi_2} e^{i\theta_\varphi}$ with $\abs{\xi_2}=\abs{\varphi_2}$, so $\theta_\xi$ and $\theta_\varphi$ represent the position of the vectors on the parallel of latitude.

By applying Definition \ref{defncauscone} to the almost commutative space-time described above we obtain the following result:

\begin{theorem}[\cite{F7}]
\label{thmcausal}
Two pure states $\omega_{p,\xi}, \omega_{q,\varphi} \in \mathcal{M}(\widetilde{\mathcal{A}})$ are causally
related with $\omega_{p,\xi} \preceq \omega_{q,\varphi}$ if and only if the following conditions are
respected: 
\begin{itemize}\itemsep=0pt
\item $p \preceq q$ in ${\mathbb R}^{1,1}$; \item $\xi$ and $\varphi$ have the same latitude;
\item $l(\gamma) \geq \frac{\abs{\theta_\varphi-\theta_\xi}}{\abs{d_1-d_2}}$, where $l(\gamma)$ represents
the length of a~causal curve $\gamma$ going from $p$ to $q$ on~${\mathbb R}^{1,1}$.
\end{itemize}
\end{theorem}

This theorem has various interesting implications concerning the causal relations in the almost commutative space-time:
\begin{itemize}
\item First of all, the usual causal relations in the space-time are preserved. This implies that the coupling with the internal space does not involve any causality violation.

\item Secondly, the motion within the internal space is only possible on a parallel of latitude, i.e.~no information can be transmitted from one parallel of latitude to another. As a consequence, no movement in the internal space is possible at the poles (note that the position of the poles and the parallels of latitude are actually determined by a unitary transformation diagonalising the finite Dirac operator). Even if such a restriction seems rather surprising, it is in fact completely coherent with a result in \cite{IochKM} where it is shown, using the Riemannian distance formula \eqref{eqnconnesdist}, that different parallels of latitude of $P(\A_F)$ are separated by an infinite distance.

\item The last condition of Theorem \ref{thmcausal} implies the existence of a strong relation between the space-time and the internal space. Indeed, the possible motion along a parallel of latitude are correlated with the length $l(\gamma) = \int \sqrt{-g_\gamma(\dot \gamma(t),\dot \gamma(t))} \,{\rm  d}t $, which physically represents the proper time of a signal moving along the curve $\gamma$. If $l(\gamma) \neq 0$, the constraint can also be understood as:
$$\frac{\abs{\theta_\varphi-\theta_\xi}}{l(\gamma)} = \frac{\text{dist}_{\text{parallel}}(\xi,\varphi)}{\text{proper time}(\gamma)} \leq \text{constant} \abs{d_1-d_2}.$$
This means that there exists an upper bound on the speed within the internal space which only depends on the values of the finite Dirac operator $D_F$. Such a constraint implies that the ratio between the distance travelled along the parallel of latitude and the time spent on the space-time cannot go beyond a fixed limit.
\end{itemize}

\begin{figure}[ht] \centering
\includegraphics[width=8cm]{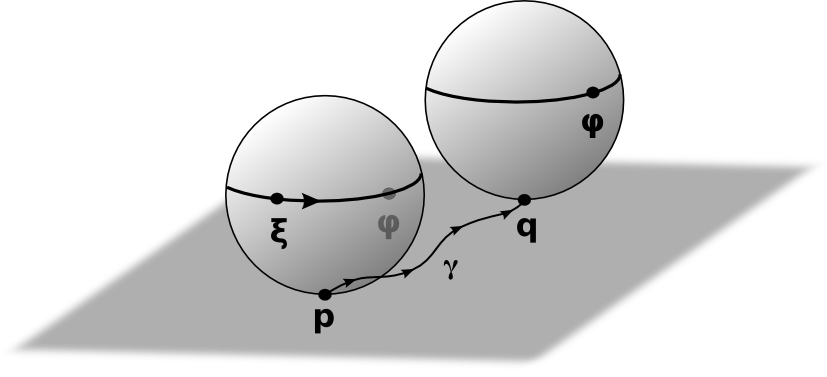}
\caption{Representation of a~path in the space of pure states.}
\label{figcausalrel}
\end{figure}

Let us take a closer look at the geometry of the almost commutative space-time considered above. We stressed at the beginning of this section that the correspondence $P(\M_2(\setC)) \cong S^2$ is purely topological and does not take into account the metric aspect. In fact, as Connes' formula \eqref{defnewdist} may take an infinite value (it is a pseudo-distance rather than a distance), it may affect also the topology. Therefore, we should rather write $P(\M_2(\setC)) \cong \bigsqcup_{i \in [0,1]} S^1$.

Let us now fix a parallel of latitude on $S^2$ and consider the following product manifold ${\mathbb R}^{1,1} \times S^1$ endowed with the metric $ds^2 = -dt^2 + dx^2 + \frac{1}{\abs{d_1-d_2}} d\theta^2$. If we consider a curve $\tilde \gamma = \left( \gamma, \gamma_{S^1} \right)$ where $\gamma$ is a curve on ${\mathbb R}^{1,1}$ going from $p$ to $q$ and $\gamma_{S^1}$ is a curve on $S^1$ going from $\theta_\xi$ to $\theta_\varphi$, then $\tilde \gamma$ is causal if and only if $\gamma$ is causal on ${\mathbb R}^{1,1}$ and $ \frac{\abs{\theta_\varphi-\theta_\xi}}{\abs{d_1-d_2}} \geq l(\gamma)$, which gives the same condition as in Theorem \ref{thmcausal}.

The geometric picture that emerges is the following: For a finite spectral triple $(\A_F,\H_F,D_F)$, the space of pure states $P(\A_F)$ equipped with Connes' pseudo-distance can be written as a disjoint sum of metric spaces. Let us stress that the final object depends on the exact form of $D_F$. For instance, if one would take $D_F$ with $d_1 = d_2$ in the above model, one would discover that all of the points in $S^2$ would be separated by an infinite distance (see \cite{F7} for details). In the almost commutative product the resulting physical space of states is just a Cartesian product of a manifold and $P(\A_F)$ with a simple coupling of the metrics.

Our results show that if the manifold is chosen to be a globally hyperbolic one, this picture remains true. We have demonstrated it on a toy model, but we conjecture that the described situation is general.

\section{Conclusions and outlook}

In this paper we have presented mathematically rigorous results on the extension of the causal structure to noncommutative geometry. Surprisingly enough, even for a simple, almost commutative toy model, the causal relations turned out to be more robust then the classical ones. However, one might ask whether the obtained mathematical structure indeed describes an aspect of the Universe? One should be aware of conceptual problems lurking in every unification theory such as noncommutative geometry. From the mathematically sophisticated theories of General Relativity and Quantum Mechanics we have learned that one should be ready to abandon one's intuitions and prejudices about such notions as time or observable. Indeed, we have seen that when spaces become noncommutative, the notion of a point -- fundamental from the point of view of physical idealisations -- becomes much harder to grasp.

Therefore, one should ask: what does causality really mean when geometry becomes noncommutative? Our understanding is that it puts constraints on the evolution of a physical system described via its states. By extending the basic Einstein's postulate that no physical signal can exceed the speed of light, we claim that no physical process can lead to an evolution in the space of states that violates the conditions imposed by the generalised causal structure. This interpretation may lead to strong physical predictions as it indicates which processes are excluded and which are allowed. These on the other hand can be explicitly tested in a suitable real-world experiment.

As mentioned in the introduction, the almost commutative geometry is uti\-lised to model fundamental interactions at a classical (not quantum) level and this is the domain in which one should seek potential physical consequences of generalised causal relations. If one chooses the algebra $\widetilde{\mathcal{A}}_F = \mathbb{C} \oplus M_2(\mathbb{C}) \oplus M_3(\mathbb{C}) \oplus \mathbb{C}$ one obtains $\P(\widetilde{\A}_F) \cong \{p_1\} \cup \{p_2\} \cup S^2 \cup \mathbb{C}P^2$. Although the model at hand is classical it is legitimate to regard $\P(\widetilde{\A}_F)$ as the space  of the corresponding quantum system \cite{rovelli}.
In this spirit, $\P(\widetilde{\A}_F)$ can be given the following interpretation: 
\begin{itemize}
\item The points $p_1$ and $p_2$ correspond to electroweak singlets -- right-handed electron and right-handed neutrino respectively.
\item The $S^2$ component encompasses the electroweak doublet consisting of left-handed electron and a left-handed neutrino.
\item $\mathbb{C}P^2$ part encodes the strong quark triplet. Note that indeed the strong interactions do not distinguish between right- and left-handed particles.
\end{itemize}

Let us also mention that the second $\mathbb{C}$ summand in $\widetilde{\mathcal{A}}_F$ allows for massive neutrinos, which implies the existence of a right-handed neutrino \cite{Christoph2}. This is consistent with the fact the $P(\mathbb{C}) = \{p\}$ adds a new single point to $\P(\widetilde{\A}_F)$.

By sticking to the particle-physics interpretation of the considered toy model we could give a free rein to our imagination and try to venture some possible physical consequences.

The first feature that appears to be general is that if one moves in the space-time with the speed of light then no change at all is allowed in the internal space. This is because the proper time along the lightlike curve is constant and equal to 0 at each moment of the evolution. One might thus speculate that any particle which is able to interact with others via a non-gravitational force must have a non-zero mass, so that it travels at a speed smaller than the speed of light. In particular, it would imply that neutrinos in \emph{all} generations are massive, since the space of states description does not distinguish between the generations. Let us stress that current experiments on neutrino oscillations imply that at least two of the neutrinos should have non-zero masses. The same reasoning might be applied to some NCG-based dark matter models (see \cite{Christoph3} for instance).

Secondly, it is natural to suspect that whenever that Riemannian distance between two states of the system is infinite, no causal relation between them can take place. This might potentially have deep physical consequences concerning the evolution of one state into another. However, to understand this aspect one would need to introduce gauge bosons, which are described by the fluctuations of the Dirac operator \cite[Section 3.5]{Chamseddine:2006ep}. It might change the picture as for instance the introduction of the Higgs field renders finite the distance between the two copies of space-time in the two-sheeted model \cite[p. 43]{Dungen}.

To draw concrete physical conclusions, one should enrich several aspects of the described toy model:
\begin{enumerate}

\item The base manifold should be a 4-dimensional, possibly curved, globally hyperbolic space-time. In our opinion this is only a technical issue and should not induce essentially new physical effects.

\item The finite algebra should be extended to the full Standard Model algebra $\widetilde{\A}_F$. This step should be performed with care, as working directly with $\widetilde{\A}_F$ seems to be too involved at the present stage. In a forthcoming paper \cite{F8} we will describe the causal structure of a two-sheeted model based on a finite algebra $\A_F = \mathbb{C} \oplus \mathbb{C}$ with a non-diagonal Dirac operator $D_F$.

\item The impact of the fluctuations of the Dirac operator on the causal structure needs to be studied. As mentioned, this operation may change the causal relations between the states rather drastically. 

\end{enumerate}

We are of course aware that the above interpretations are highly speculative and a lot of hard mathematical work needs to be done before one would be entitled to draw any conclusions about the structure of the Universe. On the other hand, the almost commutative models seem to be interesting from the point of view of physics as they might lead to testable predictions.

We should also note, that there is another approach to causality in the noncommutative regime which leads to rather different, though no less intriguing, effects \cite{Besnard2,Besnard3}. 

Finally, let us make a disclaimer about our choice of considering states as a generalisation of space-time events. Our attitude is based on the conviction that the notion of a point in noncommutative geometry simply does not make sense. The fact that the use of states or maximal ideals or irreducible representation does not necessarily lead to the same structure just means that the noncommutative geometry is more robust than the commutative one and therefore it carries more information.

To conclude this essay, let us make an outlook towards other possible directions of development of the Lorentzian noncommutative geometry. First of all one should agree on the actual definition of a Lorentzian spectral triple. To this end, one would need to establish a reconstruction theorem allowing to identify the exact class of Lorentzian spectral triples corresponding to Lorentzian manifolds, or at least to globally hyperbolic space-times. An extension of Connes' reconstruction theorem to pseudo-Riemannian or Lorentzian manifolds, which is deep in the shadows at this moment, is a hope for the future.

The second significant drawback of Lorentzian noncommutative geometry is that only two out of three steps needed in order to set a physical model (the topology and the differential structure) are partially fulfilled. The question of the dynamics is still completely open. The usual tools to compute the spectral action in the Riemannian setting are not valid, because of the vicious properties of the Lorentzian Dirac operator. It is not even known whether the problem only requires new techniques or whether one should modify the very paradigm of the spectral action principle. This is probably the biggest challenge for the future of Lorentzian noncommutative geometry.

\section*{Acknowledgement}

This work was supported by a grant from the John Templeton Foundation. M.E. would like to thank Fabien Besnard and Christoph Stephan for the stimulating discussions during the FFP14 conference.


\end{document}